\documentclass{article}
\usepackage{multirow, amsmath, amsthm, amssymb}
\usepackage[margin = 2cm]{geometry}
\usepackage[colorlinks = true]{hyperref}

\newtheorem{definition}{Definition}
\newtheorem{lemma}{Lemma}
\newtheorem{theorem}{Theorem}
\newtheorem{corollary}{Corollary}

\DeclareMathOperator{\Ln}{Ln}

\title{A two-parameter entropy and its fundamental properties}
\author{Supriyo Dutta\\
	\small{Centre for Theoretical Studies, Indian Institute of Technology Kharagpur,}\\
	\small{Kharagpur, West Bengal, India - 721302.} \\
	\small{Email: \texttt{dosupriyo@gmail.com}} \vspace{.25cm} \\
	Shigeru Furuichi\\
	\small{Department of Information Science, College of Humanities and Sciences, Nihon University,}\\
	\small{3-25-40, Sakurajyousui, Setagaya-Ku, Tokyo, 156-8550, Japan.}\\
	\small{Email: \texttt{furucihi@chs.nihon-u.ac.jp}}  \vspace{.25cm} \\
	Partha Guha\\
	\small{Department of Mathematics, Khalifa University,}\\
	\small{Zone 1 - Abu Dhabi, United Arab Emirates.}\\
	\small{Email: \texttt{partha.guha@ku.ac.ae}}
} 
\date{} 
\begin{document}

	\maketitle
	
	\begin{abstract}
		\noindent This article proposes a new two-parameter generalized entropy, which can be reduced to the Tsallis and the Shannon entropy for specific values of its parameters. We develop a number of information-theoretic properties of this generalized entropy and divergence, for instance, the sub-additive property, strong sub-additive property, joint convexity, and information monotonicity. This article presents an exposit investigation on the information-theoretic and information-geometric characteristics of the new generalized entropy and compare them with the properties of the Tsallis and the Shannon entropy.\\
		\noindent\textbf{keywords:} Deformed logarithm; Tsallis entropy; relative entropy; chain rule; sub-additive property; information geometry.
	\end{abstract}

	\section{Introduction}
	
		We encounter complex systems obeying asymptotic power-law distributions in different fields of science and technology. For explaining the statistical natures of these complex systems, an effective approach is addressing statistical mechanics in the form of a suitable generalization of the Shannon entropy. The Tsallis’ non-extensive thermostatistics \cite{tsallis1988possible} is one of such generalizations, which is utilized in image processing \cite{de2004image}, medical engineering \cite{zhang2009application}, signal analysis \cite{chen2014tsallis}, quantum information \cite{becker2019convergence, abe2002towards}, and in many other disciplines, in the recent years. The Sharma-Mittal entropy \cite{sharma1975entropy, mittal1975some} is a two-parameter generalization of the Shannon entropy which incorporates a large number of prominent entropy measures as special cases, such as the Tsallis and R{\'e}nyi entropy. It is useful in the investigations of diffusion processes in statistical physics \cite{frank2000exact}, analysis of record values in statistics \cite{paul2016sharma}, estimating the performance of clustering models in data analysis\cite{koltcov2019estimating}, and modeling uncertainty in the theory of human cognition \cite{crupi2018generalized}. In the context of astrophysics, generalized entropy is useful in modeling holographic dark energy \cite{jahromi2018generalized, younas2019cosmological}, and in the investigation of the different phenomenon of black holes \cite{sadeghi2019investigation, ghaffari2019black}. 
		
		This article concentrates on the information theoretic properties of a generalized entropy with two parameters. In the literature, a number of two-parameter generalized entropy are proposed in the context of thermodynamics and statistical mechanics. Given a discrete probability distribution $\mathcal{P} = \{p(x): x \in X\}$, the Sharma-Mittal entropy \cite{sharma1975entropy, mittal1975some} of a random variable $X$ is defined by 
		\begin{equation}
		SM_{\{\alpha, \beta\}}(X) = \frac{1}{\beta - 1} \left[1 - \left( \sum_{x \in X} \left(p(x)\right)^\alpha\right)^{\frac{1 - \beta}{1 - \alpha}}\right],
		\end{equation}
		for two real parameters $\alpha \neq 1$ and $\beta \neq 1$. Another two-parameter entropy was defined by Borges and Roditi \cite{borges1998family} which is  
		\begin{equation}
		BR_{\{\alpha, \beta\}}(X) = \sum_{x \in X} \frac{(p(x))^\alpha - (p(x))^\beta}{\beta - \alpha},
		\end{equation}
		where $\alpha \neq \beta$. Later in \cite{kaniadakis2004deformed, kaniadakis2005two} a two-parameter entropy was proposed by Kaniadakis, Lissia, and Scarfone, which is
		\begin{equation}\label{Original_definition_of_SM_entropy}
		KLS_{\{k,r\}}(X) = \sum_{x \in X} \left(p(x) \right)^{1 + r} \frac{(p(x))^k - (p(x))^{-k}}{2k} = -\sum_{x \in X} p(x) \Ln_{\{k,r\}}\left(p(x)\right), 
		\end{equation}
		where $\Ln_{\{k,r\}}(u) = u^r \frac{u^k - u^{-k}}{2k}$ and the parameters $k$ and $r$ were chosen from $\mathcal{R} = \{(k, r): -|k| \leq r \leq |k|, 0 < |k| < \frac{1}{2}\} \cup \{(k, r): |k| - 1 \leq r \leq 1 - |k|, \frac{1}{2} \leq |k| < 1\}$. The information theoretic properties of $KLS_{\{k, r\}}$ and $BR_{\{\alpha, \beta\}}$ are investigated in \cite{naudts2002deformed}, and  \cite{wada2007two, furuichi2010axiomatic}, respectively.
		
		We observe that a modification to the parameters $k$ and $r$ of $\Ln_{\{k,r\}}$ provides a product rule of the two parameter deformed logarithm. It leads us to define the two-parameter generalized entropy $S_{\{k,r\}}$ and the generalized divergence $D_{\{k,r\}}$. The significant attributes of $S_{\{k,r\}}$ and $D_{\{k,r\}}$ derived in this article are listed below:
		\begin{enumerate}
			\item 
			The pseudo-additivity of $S_{\{k,r\}}$ (Equation (\ref{pseudo_additive_property_of_SM})): Given any two discrete random variables $X$ and $Y$ we have 
			\begin{equation}
			S_{\{k,r\}}(X, Y) = S_{\{k,r\}}(X) + S_{\{k,r\}}(Y) - 2k S_{\{k,r\}}(X) S_{\{k,r\}}(Y).
			\end{equation}
			\item 
			The sub-additive property of $S_{\{k,r\}}$ (Theorem \ref{sub_additive_property}) : Given a sequence of random variables $X_1, X_2, \dots X_n$, it can be proved that
			\begin{equation}
			S_{\{k,r\}}(X_1, X_2, \dots X_n) \leq \sum_{i = 1}^n S_{\{k,r\}}(X_i).
			\end{equation}
			\item 
			The pseudo-additivity of $D_{\{k,r\}}$ (Theorem \ref{pseudo_additivity_of_divergence}): Consider probability distributions $\mathcal{P}^{(1)}$, and $\mathcal{Q}^{(1)}$ defined on a random variable $X$ as well as $\mathcal{P}^{(2)}$, and $\mathcal{Q}^{(2)}$ defined on random variable $Y$. Then,
			\begin{equation}
				D_{\{k,r\}}(\mathcal{P}^{(1)} \otimes \mathcal{P}^{(2)} || \mathcal{Q}^{(1)} \otimes \mathcal{Q}^{(2)}) =  D_{\{k,r\}}(\mathcal{P}^{(1)} || \mathcal{Q}^{(1)}) + D_{\{k,r\}}(\mathcal{P}^{(2)} || \mathcal{Q}^{(2)}) - 2k D_{\{k,r\}}(\mathcal{P}^{(1)} || \mathcal{Q}^{(1)}) D_{\{k,r\}}(\mathcal{P}^{(2)} || \mathcal{Q}^{(2)}).			
			\end{equation}
			\item 
			The joint convexity of $D_{\{k, r\}}$ (Theorem \ref{joint_convexity_of_divergence}):
			\begin{equation}
			D_{\{k, r\}}(\mathcal{P}^{(1)} + \lambda \mathcal{P}^{(2)} || \mathcal{Q}^{(1)} + \lambda \mathcal{Q}^{(2)}) \leq D_{\{k, r\}}(\mathcal{P}^{(1)} || \mathcal{Q}^{(1)}) + \lambda D_{\{k, r\}}(\mathcal{P}^{(2)} || \mathcal{Q}^{(2)}).
			\end{equation}
			\item 
			The information monotonicity of $D_{\{k, r\}}$ (Theorem \ref{information_monotonicity}) : Given any two probability distributions $\mathcal{P}$ and $\mathcal{Q}$ of a random variable and a probability transition matrix $W$ we have
			\begin{equation}
			D_{\{k,r\}}(W\mathcal{P}|| W\mathcal{Q}) \leq D_{\{k,r\}}(\mathcal{P}|| \mathcal{Q}).
			\end{equation}
		\end{enumerate}
		The similar properties for the Tsallis entropy and divergence are investigated in detail \cite{furuichi2004fundamental},  \cite{furuichi2006information}, \cite{furuichi2005uniqueness}. To the best of our knowledge, this article develops these properties for two-parameter generalized entropy first time in literature.
		
		This article is distributed as follows. In section 2, we define the joint entropy and the conditional entropy to present a number of properties of two-parameter generalized entropy as well as the chain rule. Section 3 is dedicated to two-parameter generalized relative entropy and its properties. We discuss the information geometric aspects of entropy in section 4. Then we conclude the article comparing similar properties of Shannon, Tsallis and two-parameter generalized entropy.

	\section{Two-parameter generalized entropy}
	
		From classical information theory we recall that the function $f(u) = -\log(u)$ is a positive, monotone decreasing, convex function where $0 \leq u \leq 1$ where the convention $0 \log 0 = 0$ is used. The two-parameter deformed logarithm should preserve equivalent properties. Below, we define a two parameter deformed logarithm justify its characteristics.  
		\begin{definition}\label{redefined_lnkr}
			$$\ln_{\{k,r\}}(u) = \frac{u^{k} - u^{-k}}{2k u^{r}} = \frac{u^{2k} - 1}{2k u^{r + k}},$$
			with $r > 0$ and $0 < k \leq 1$. 
		\end{definition}
		
		\begin{lemma}\label{convexity_of_lnkr}
			For $r < 0$, and $0 < k \leq 1$ the function $-\ln_{\{k,r\}}(u) = -u^r \frac{u^{k} - u^{-k}}{2k}$ is positive, convex, and monotonically decreasing for all $u \in (0, 1]$.
		\end{lemma}
		
		\begin{proof}
			Recall that a twice differentiable function $f(u), u \in \mathbb{R}$ is convex if $f''(u) > 0$. Note that, $f(u) = u^r$ is a positive, monotone decreasing, and convex function for all $u \in (0, 1]$ and $r < 0$. Also, for all $k > 0$ and $u \in (0, 1]$ we have $u^{-k} \geq u^{k}$. Therefore, the function $g(u) = - \frac{u^k - u^{-k}}{2k}$ is a positive and monotone decreasing function. For convexity, we need $g''(u) = - \frac{(k - 1)u^{k - 2} - (k + 1) u^{-k - 2}}{2} \geq 0$, which holds for $0 < k \leq 1$. We know that, if two given functions $f, g: \mathbb{R} \rightarrow \mathbb{R}^{+}$ are convex, and both monotonically decreasing  on an interval, then $fg(u) = f(u)g(u)$ is convex \cite{boyd2004convex}. Combining we get $-\ln_{\{k,r\}}(u) = f(u) g(u)$ is a positive, monotonically decreasing, and convex function.
		\end{proof}	
		
		In the next lemma, we present a product rule for $\ln_{\{k, r\}}$ which leads us to the chain rule of generalized entropy.
		
		\begin{lemma}\label{product_1}
			Given any two real numbers $u, v \neq 0$ we have
			$$(uv)^{r + k} \ln_{\{k,r\}}(uv) = u^{r + k} \ln_{\{k,r\}}(u) + v^{r + k} \ln_{\{k,r\}}(v) + 2k u^{r + k} v^{r + k} \ln_{\{k,r\}}(u)\ln_{\{k,r\}}(v).$$
		\end{lemma}
		\begin{proof}
			\begin{equation}
				\begin{split}
					\ln_{\{k,r\}}(u)\ln_{\{k,r\}}(v) & = \frac{u^{2k} - 1}{2k u^{r+k}}\frac{v^{2k} - 1}{2k v^{r+k}} = \frac{u^{2k}v^{2k} - u^{2k} - v^{2k} + 1}{4k^2 u^{r + k}v^{r + k}} \\
					& = \frac{u^{2k}v^{2k} -1 + 1- u^{2k} - v^{2k} + 1}{4k^2 u^{r + k}v^{r + k}} \\
					& = \frac{u^{2k}v^{2k} -1}{4k^2 u^{r + k}v^{r + k}} - \frac{u^{2k} - 1}{4k^2 u^{r + k}v^{r + k}} - \frac{v^{2k} - 1}{4k^2 u^{r + k}v^{r + k}}\\
					& =\frac{\ln_{\{k,r\}}(uv)}{2k} - \frac{\ln_{\{k,r\}}(u)}{2k v^{r + k}} - \frac{\ln_{\{k,r\}}(v)}{2k u^{r + k}}.
				\end{split}
			\end{equation}
			Simplifying, we get the result.
		\end{proof}
		
		Note that, in Lemma \ref{product_1} every term of $\ln_{\{k,r\}}(z)$ has the coefficient $z^{r + k}$ for $z = u$ and $v$. This structure motivates us to keep a term of $z^{r + k}$ with $\ln_{\{k, r\}}(z)$ in definition of entropy. Hence, we define the two-parameter generalized entropy as follows: 
		
		\begin{definition}\label{modified_definition_of_SM_entropy}
			We define the two-parameter generalized entropy for a random variable $X$ with probability distribution $\mathcal{P} = \{p(x)\}_{x \in X}$ as 
			$$S_{\{k,r\}}(X) = - \sum_{x \in X} \left(p(x)\right)^{r + k + 1} \ln_{\{k,r\}}(p(x)),$$
			where $\ln_{\{k,r\}}(u) = \frac{u^{k} - u^{-k}}{2k u^{r}}$ with $0 < k \leq \frac{1}{2}$, and $r > 0$.
		\end{definition}
		
		In Definition \ref{modified_definition_of_SM_entropy}, if $p(x) = 0$ for some $x \in X$ then conventionally we have 
		$$0^{r + k + 1}\ln_{\{k, r\}}(0) = \lim_{p(x) \rightarrow 0}\left(p(x)\right)^{r + k + 1} \ln_{\{k,r\}}(p(x)) = 0.$$ Here, restriction in the domain of $k$ is essential for proving Lemma \ref{inequality_for_subadditive_property}  and \ref{inequality_for_strongly_sub_additivity_property}. Lemma \ref{convexity_of_lnkr} suggests that for any random variable $X$ we have $S_{\{k,r\}}(X) \geq 0$. Moreover, $S_{\{k, r\}}$ reduces to the Tsallis entropy when $k = r = \frac{q - 1}{2}$ that is 
		\begin{equation}
		S_{\left\{\frac{q - 1}{2}, \frac{q - 1}{2} \right\}}(X) = - \sum_{x \in X} \left(p(x)\right)^q \frac{(p(x))^{1 - q} - 1 }{1 - q} = S_q(X).
		\end{equation}
		
		An alternative expression of $S_{\{k, r\}}$ can be presented. We can verify that
		\begin{equation}
		\ln_{\{k,r\}}(uv) = \frac{1}{u^{r - k}}\ln_{\{k,r\}}(v) + \frac{1}{v^{r + k}} \ln_{\{k,r\}}(u).
		\end{equation}
		Putting $v = \frac{1}{u}$ in this equation we find
		\begin{equation}\label{inversion_under_logarithm}
		\ln_{\{k,r\}}\left(\frac{1}{v}\right) = - u^{2r}\ln_{\{k,r\}}(u), ~\text{or}~ \ln_{\{k,r\}}(u) = - \frac{1}{u^{2r}}\ln_{\{k,r\}}\left(\frac{1}{u}\right).
		\end{equation}
		Therefore, Definition \ref{modified_definition_of_SM_entropy} suggests that 
		\begin{equation}
		S_{\{k,r\}}(X) = \sum_{x \in X} \left(p(x)\right)^{k - r + 1} \ln_{\{k,r\}} \left( \frac{1}{p(x)} \right).
		\end{equation}
		
		\begin{definition}\label{joint_entropy}
			(Joint entropy) Let $\mathcal{P} = \{p(x, y)\}_{(x,y) \in (X,Y)}$ be a probability distribution of the joint random variable $(X, Y)$. The generalized joint entropy of $(X, Y)$ is defined by 
			$$S_{\{k,r\}}(X, Y) = - \sum_{x \in X} \sum_{y \in Y} \left(p(x,y)\right)^{k + r + 1} \ln_{\{k,r\}}(p(x,y)).$$
		\end{definition}
		
		Similarly, for three random variables $X, Y$, and $Z$ the joint entropy is
		\begin{equation}
		S_{\{k,r\}}(X, Y, Z) = - \sum_{x \in X} \sum_{y \in Y} \sum_{z \in Z} \left(p(x, y, z)\right)^{k + r + 1} \ln_{\{k,r\}}(p(x, y, z)).
		\end{equation}
		
		\begin{definition}\label{conditional_entropy}
			(Conditional entropy) Given a conditional random variable $Y|X = x$ we define the generalized conditional entropy as
			\begin{equation*}
			\begin{split}
			S_{\{k,r\}}(Y|X) & = \sum_{x \in X} (p(x))^{2k + 1} S_{\{k,r\}}(Y|X = x) \\
			& = - \sum_{x \in X} (p(x))^{2k + 1} \sum_{y \in Y} \left(p(y|x)\right)^{k + r + 1} \ln_{\{k,r\}}(p(y|x)) \\
			& = - \sum_{x \in X} \sum_{y \in Y} (p(x))^{2k + 1} \left( p(y|x)\right)^{k + r + 1} \ln_{\{k,r\}}(p(y|x)).
			\end{split}
			\end{equation*}
		\end{definition}
		
		As $\ln_{\{k,r\}}(u) = - \frac{1}{u^{2r}}\ln_{\{k,r\}}\left(\frac{1}{u}\right)$, we can alternatively write down
		\begin{equation}
		S_{\{k,r\}}(Y|X) = \sum_{x \in X} \sum_{y \in Y} (p(x))^{2k + 1} \left( p(y|x)\right)^{k - r + 1} \ln_{\{k,r\}}\left(\frac{1}{p(y|x)}\right).
		\end{equation} 
		
		This definition can be generalized for three or more random variables. Given three random variables $X, Y$ and $Z$ we have
		\begin{equation}
		\begin{split}
		S_{\{k,r\}}(X, Y| Z) & = - \sum_{x \in X} \sum_{y \in Y} \left( p(z) \right)^{2k + 1} S_{\{k, r\}} (X, Y| Z = z)\\
		& = - \sum_{x \in X} \sum_{y \in Y} \sum_{z \in Z} \left( p(z) \right)^{2k + 1} \left( p(x, y | z) \right)^{ k + r + 1} \ln_{\{k, r\}} \left( p(x, y| z) \right).
		\end{split}
		\end{equation}
		In a similar fashion, we can define
		\begin{equation}
		\begin{split}
		S_{\{k,r\}}(Y| X, Z) & = - \sum_{x \in X} \sum_{y \in Y} \left( p(x, z) \right)^{2k + 1} S_{\{k, r\}} (Y| X = x, Z = z)\\
		& = - \sum_{x \in X} \sum_{y \in Y} \sum_{z \in Z} \left( p(x,z) \right)^{2k + 1} \left( p(y | x, z) \right)^{ k + r + 1} \ln_{\{k, r\}} \left( p(y| x, z) \right).
		\end{split}
		\end{equation}
		Likewise, definition of the conditional entropy can be extended for any number of random variables for defining $S_{\{k,r\}}(X_1, X_2, \dots X_n| Y_1, Y_2, \dots Y_m)$. Now we prove a number of characteristics of generalized entropy.
		
		\begin{lemma}\label{rule_for_conditional_entropy}
			Given two independent random variables $X$ and $Y$ the generalized conditional entropy can be expressed as 
			$$S_{\{k,r\}}(Y | X) = S_{\{k,r\}}(Y) - 2S_{\{k,r\}}(X)S_{\{k,r\}}(Y).$$
		\end{lemma}
		
		\begin{proof}
			Definition of $\ln_{\{k, r\}}$ suggests that $(p(x))^{2k} = 1 + 2k (p(x))^{r + k} \ln_{\{k,r\}}(p(x))$. Putting it in definition of the conditional entropy we construct
			\begin{equation}
			S_{\{k,r\}}(Y|X) = - \sum_{x \in X} (p(x))\left[ 1 + 2k (p(x))^{r + k} \ln_{\{k,r\}}(p(x)) \right] \sum_{y \in Y} \left(p(y|x)\right)^{r + k + 1} \ln_{\{k,r\}}(p(y|x)).
			\end{equation}
			As $X$ and $Y$ are independent we have $p(y|x) = p(y)$. Therefore,
			\begin{equation}
				\begin{split}
					& S_{\{k,r\}}(Y|X) = - \sum_{x \in X} (p(x))\left[ 1 + 2k (p(x))^{r + k} \ln_{\{k,r\}}(p(x)) \right]  \times \sum_{y \in Y} \left(p(y)\right)^{r + k + 1} \ln_{\{k,r\}}(p(y)) \\
					= & - \sum_{x \in X} (p(x)) \sum_{y \in Y} \left(p(y)\right)^{r + k + 1} \ln_{\{k,r\}}(p(y)) - \sum_{x \in X} 2k (p(x))^{r + k + 1} \ln_{\{k,r\}}(p(x)) \sum_{y \in Y} \left(p(y)\right)^{r + k + 1} \ln_{\{k,r\}}(p(y))\\
					= & S_{\{k,r\}}(Y) - 2k S_{\{k,r\}}(X)S_{\{k,r\}}(Y).
				\end{split}
			\end{equation}
		\end{proof}
		
		Lemma \ref{rule_for_conditional_entropy} suggests that $S_{\{k, r\}}(Y | X) \leq S_{\{k, r\}}(Y)$ for independent random variables $X$ and $Y$. The next lemma proves this inequality for any two random variables.
		
		\begin{lemma}\label{inequality_for_subadditive_property} 
			Given any two random variables $X$ and $Y$ we have $S_{\{k, r\}}(Y | X) \leq S_{\{k, r\}}(Y)$.
		\end{lemma}
		
		\begin{proof}
			Note that, the function $f(u) = u^{k + r + 1} \ln_{\{k, r\}}(u)$ where $r > 0, 0 < k \leq \frac{1}{2}$ and $0 \leq u \leq 1$ is a convex function, that is $-f(u)$ is a concave function. As $0 \leq p(x) \leq 1$, we have $0 \leq \left( p(x) \right)^{2k + 1} \leq p(x) \leq 1$. Also, $0 \leq p(y|x) \leq 1$ indicates $-f(p(y|x)) = - \left(p(y|x) \right)^{k + r + 1} \ln_{\{k, r\}} \left( p(y | x) \right) = \left(p(y|x) \right)^{k - r + 1} \ln_{\{k, r\}} \left( \frac{1}{p(y | x)} \right) \geq 0$, for $0 \leq x \leq 1$. Combining we get 
			\begin{equation}
			- (p(x))^{2k + 1} \left(p(y|x) \right)^{k + r + 1} \ln_{\{k, r\}} \left( p(y | x) \right) \leq - p(x) \left(p(y|x) \right)^{k + r + 1} \ln_{\{k, r\}} \left( p(y | x) \right).
			\end{equation}
			Now, applying the concavity property of $-f(u)$ we find
			\begin{equation}\label{critical_inequality}
			- \sum_{x \in X} p(x) f \left( p(y|x) \right) \leq - f \left( \sum_{x \in X} p(x) p(y|x) \right) = - f \left( \sum_{x \in X} p(x, y) \right) = - f\left( p(y) \right).
			\end{equation} 
			Expanding $f(p(y|x))$ in the above equation,
			\begin{equation}
			- \sum_{x \in X} p(x) \left(p(y|x) \right)^{k + r + 1} \ln_{\{k, r\}} \left( p(y | x) \right) \leq - \left(p(y)\right)^{k + r + 1} \ln_{\{k, r\}} \left(p(y)\right).
			\end{equation} 
			Summing over $Y$ we find
			\begin{equation} 
			- \sum_{x \in X} p(x) \sum_{y \in Y} \left(p(y|x) \right)^{k + r + 1} \ln_{\{k, r\}} \left( p(y | x) \right) \leq - \sum_{y \in Y} \left(p(y)\right)^{k + r + 1} \ln_{\{k, r\}} \left(p(y)\right).
			\end{equation}
			Combining this equation with equation (\ref{critical_inequality}) we find
			\begin{equation}
			\begin{split}
			- \sum_{x \in X} (p(x))^{2k + 1} \sum_{y \in Y} \left(p(y|x) \right)^{k + r + 1} \ln_{\{k, r\}} \left( p(y | x) \right) 
			\leq & - \sum_{x \in X} p(x) \sum_{y \in Y} \left(p(y|x) \right)^{k + r + 1} \ln_{\{k, r\}} \left( p(y | x) \right)\\
			\leq & - \sum_{y \in Y} \left(p(y)\right)^{k + r + 1} \ln_{\{k, r\}} \left(p(y)\right).
			\end{split}
			\end{equation} 
			The first and the last term of the above inequality indicates $S_{\{k, r\}}(Y | X) \leq S_{\{k, r\}}(Y)$. 
		\end{proof}
		
		\begin{theorem}\label{chain_rule_for_Sharma_MIttal_entropy}
			(Chain rule for generalized entropy) Given any two random variables $X$ and $Y$ we have 
			$$S_{\{k,r\}}(X, Y) = S_{\{k,r\}}(X) + S_{\{k,r\}}(Y|X).$$
		\end{theorem}
		
		\begin{proof}
			The product rule of $\ln_{\{k,r\}}(u)$ mentioned in Lemma \ref{product_1} indicates that  
			\begin{equation}
				\begin{split}
					(p(x)p(y|x))^{r + k} \ln_{\{k,r\}}(p(x)p(y|x)) = & p(x)^{r + k} \ln_{\{k,r\}}(p(x)) + p(y|x)^{r + k} \ln_{\{k,r\}}(p(y|x)) \\
					& \hspace{2cm} + 2k p(x)^{r + k} p(y|x)^{r + k} \ln_{\{k,r\}}(p(x))\ln_{\{k,r\}}(p(y|x)).
				\end{split} 
			\end{equation}
			Applying $p(x,y) = p(x) p(y|x)$ we find that 
			\begin{equation} 
				\begin{split} 
					(p(x,y))^{r + k} \ln_{\{k,r\}}(p(x,y)) = & p(x)^{r + k} \ln_{\{k,r\}}(p(x)) + p(y|x)^{r + k} \ln_{\{k,r\}}(p(y|x))\\
					& + 2k p(x)^{r + k} p(y|x)^{r + k} \ln_{\{k,r\}}(p(x))\ln_{\{k,r\}}(p(y|x))\\
					= & p(x)^{r + k} \ln_{\{k,r\}}(p(x)) + [1 + 2k (p(x))^{r + k}\ln_{\{k,r\}}(p(x))] p(y|x)^{r + k} \ln_{\{k,r\}}(p(y|x)). \\
				\end{split}
			\end{equation}
			Definition \ref{modified_definition_of_SM_entropy} of the generalized entropy suggests that $(p(x))^{2k} = 1 + 2k (p(x))^{r + k} \ln_{\{k,r\}}(p(x))$. Putting it in the above equation we find
			\begin{equation}
			(p(x,y))^{r + k} \ln_{\{k,r\}}(p(x,y)) = p(x)^{r + k} \ln_{\{k,r\}}(p(x)) + (p(x))^{2k} p(y|x)^{r + k} \ln_{\{k,r\}}(p(y|x)).
			\end{equation}
			Multiplying both side by $p(x,y)$ and summing over $X$ and $Y$ we get 
			\begin{equation}
				\begin{split}
					- \sum_{x \in X} \sum_{y \in Y} p(x,y))^{r + k + 1} \ln_{\{k,r\}}(p(x,y)) = & - \sum_{x \in X} \sum_{y \in Y} p(x,y) p(x)^{r + k} \ln_{\{k,r\}}(p(x)) \\
					& \hspace{2cm} - \sum_{x \in X} \sum_{y \in Y} p(x,y)(p(x))^{2k} p(y|x)^{r + k} \ln_{\{k,r\}}(p(y|x)). 
				\end{split}
			\end{equation}
			Now, definitions of the joint entropy and the conditional entropy together indicate 
			\begin{equation}
			\begin{split}
			S_{\{k,r\}}(X, Y) = & - \left[\sum_{x \in X} p(x)^{r + k + 1} \ln_{\{k,r\}}(p(x)) \right] \left[\sum_{y \in Y} p(y|x) \right] \\
			& \hspace{2cm} - \sum_{x \in X} \sum_{y \in Y} (p(x))^{2k + 1} p(y|x)^{r + k + 1} \ln_{\{k,r\}}(p(y|x)) \\
			\text{or}~ S_{\{k,r\}}(X, Y) = & S_{\{k,r\}}(X) + S_{\{k,r\}}(Y|X).\\  
			\end{split}
			\end{equation}
		\end{proof}
		
		The above theorem clearly indicates that $S_{\{k,r\}}(X) \leq S_{\{k,r\}}(X, Y)$. For two independent random variables $X$ and $Y$ Lemma \ref{rule_for_conditional_entropy} and Theorem \ref{chain_rule_for_Sharma_MIttal_entropy} produce that the pseudo-additivity property for the generalized entropy which is
		\begin{equation}\label{pseudo_additive_property_of_SM}
		S_{\{k,r\}}(X, Y) = S_{\{k,r\}}(X) + S_{\{k,r\}}(Y) - 2k S_{\{k,r\}}(X) S_{\{k,r\}}(Y).
		\end{equation}
		
		\begin{corollary}\label{chani_rule_1}
			The following chain rules holds for the generalized entropy: $S_{\{k,r\}}(X,Y,Z) = S_{\{k,r\}}(X,Y|Z) + S_{\{k,r\}}(Z)$.
		\end{corollary}
		\begin{proof}
			We have $p(x, y, z) = p(x, y | z) p(z)$. Now, applying the product rule mentioned in Lemma \ref{product_1} we find 
			\begin{equation}
			\begin{split}
			\left( p(x, y, z) \right)^{r + k} \ln_{\{k, r\}} \left( p(x, y, z) \right) & = \left( p(z) \right)^{r + k} \left( p(x, y| z) \right)^{r + k} \ln_{\{k, r\}} \left( p(x, y | z) p(z) \right)  \\
			& = \left( p(z) \right)^{ r + k} \ln_{\{k, r\}} \left( p(z) \right) + \left( p(x,y| z) \right)^{ r + k} \ln_{\{k, r\}} \left( p(x, y | z) \right) \\
			& \hspace{2cm} + 2k \left( p(z) \right)^{ r + k} \left( p(x,y| z) \right)^{ r + k} \ln_{\{k, r\}} \left( p(z) \right) \ln_{\{k, r\}} \left( p(x, y | z) \right).
			\end{split}
			\end{equation}
			Now the equation $\left( p(z) \right)^{2k} = 1 + 2k \left( p(z) \right)^{r + k} \ln_{\{k, r\}} \left( p(z) \right)$ and definitions of joint and conditional entropies indicate $S_{\{k,r\}}(X,Y,Z) =S_{\{k,r\}}(X,Y|Z) + S_{\{k,r\}}(Z)$.
		\end{proof}
		
		\begin{corollary}\label{chani_rule_2} 
			The generalized entropy also fulfills the chain rule: 
			$$S_{\{k,r\}}(X,Y|Z) = S_{\{k,r\}}(X|Z) + S_{\{k,r\}}(Y| X,Z).$$ 
		\end{corollary}
		\begin{proof}
			We also have $p(x, y, z) = p(y| x, z) p(x, z)$. Applying the similar approach in Corollary \ref{chani_rule_1} and Theorem \ref{chain_rule_for_Sharma_MIttal_entropy} we have 
			\begin{equation}
			\begin{split}
			& S_{\{k, r\}} (X, Y, Z) = S_{\{k , r\}} (Y | X, Z) + S_{\{k, r\}}(X, Z) \\ 
			\text{or}~ & S_{\{k, r\}} (Y | X, Z ) = S_{\{k, r\}} (X, Y, Z) - S_{\{k, r\}} (X, Z). 
			\end{split}
			\end{equation}
			Applying Corollary \ref{chani_rule_1} we have 
			\begin{equation}
			\begin{split}
			S_{\{k, r\}} (Y | X, Z ) = S_{\{k,r\}}(X,Y|Z) + S_{\{k,r\}}(Z) - S_{\{k, r\}} (X, Z).
			\end{split}
			\end{equation}
			Now Theorem \ref{chain_rule_for_Sharma_MIttal_entropy} suggests $S_{\{k, r\}} (X, Z) = S_{\{k, r\}} (Z) + S_{\{k, r\}} (X|Z)$. Putting it in the above equation we have
			\begin{equation}
			\begin{split}
			& S_{\{k, r\}} (Y | X, Z ) = S_{\{k,r\}}(X,Y|Z) + S_{\{k,r\}}(Z) - [S_{\{k, r\}} (Z) + S_{\{k, r\}} (X|Z)] \\
			\text{or}~ & S_{\{k, r\}} (Y | X, Z ) = S_{\{k,r\}}(X,Y|Z) - S_{\{k, r\}} (X|Z) \\
			\text{or}~ & S_{\{k,r\}}(X,Y|Z) = S_{\{k,r\}}(X|Z) + S_{\{k,r\}}(Y| X,Z). 
			\end{split}
			\end{equation}
		\end{proof}
		
		Corollary \ref{chani_rule_2} also suggests that $S_{\{k,r\}}(X|Z) \leq S_{\{k,r\}}(X,Y|Z)$. In general Corollary \ref{chani_rule_1} and \ref{chani_rule_2} can be generalized as
		\begin{equation}
		S_{\{k, r\}}(X_1, X_2, \dots X_n|Y) = \sum_{i = 1}^n S_{\{k, r\}}(X_i| X_{i - 1}, \dots , X_1, Y), 
		\end{equation}
		which indicates
		\begin{equation}
		S_{\{k, r\}}(X_1, X_2, \dots X_n) = \sum_{i = 1}^n S_{\{k, r\}} (X_i | X_{i - 1}, \dots, X_1). 
		\end{equation}
		
		For any two independent random variables $X$ and $Y$ equation (\ref{pseudo_additive_property_of_SM}) suggests that $S_{\{k,r\}}(X, Y) \leq S_{\{k,r\}}(X) + S_{\{k,r\}}(Y)$. If $X$ and $Y$ are any two random variables Theorem \ref{chain_rule_for_Sharma_MIttal_entropy} and Lemma \ref{inequality_for_subadditive_property} together indicate the following theorem, which is the sub-additive property for the generalized entropy.
		\begin{theorem}\label{sub_additive_property}
			Given any two random variables $X$ and $Y$ we have $S_{\{k,r\}}(X, Y) \leq S_{\{k,r\}}(X) + S_{\{k,r\}}(Y)$.
		\end{theorem}
		For random variables $X_1, X_2, \dots X_n$ this theorem can be further generalized as
		\begin{equation}
		S_{\{k,r\}}(X_1, X_2, \dots X_n) \leq \sum_{i = 1}^n S_{\{k,r\}}(X_i).
		\end{equation} 
		
		\begin{lemma}\label{inequality_for_strongly_sub_additivity_property}
			Given any three random variables $X$, $Y$ and $Z$ we have $S_{\{k,r\}}(Y|Z) \geq S_{\{k,r\}}(Y|X,Z)$.
		\end{lemma}
		
		\begin{proof}
			Observe that the function $f(u) = u^{k + r + 1} \ln_{\{k, r\}}(x)$, where $r > 0, 0 < k \leq \frac{1}{2}$ and $0 \leq u \leq 1$ is a convex function, as well as $f(u) \leq 0$. Therefore, as $0 \leq p(y|z) \leq 1$ we have 
			\begin{equation}
			- f(p(y|z)) = - (p(y|z))^{r + k + 1} \ln_{\{k,r\}}(p(y|z)) > 0.
			\end{equation}	 
			In addition, $0 \leq p(y|x,z) \leq 1$ indicates 
			\begin{equation}
			-p(x|z)f(p(y|x,z)) = p(x|z) (p(y|x,z))^{r + k + 1} \ln_{\{k,r\}}(p(y|x,z)) \geq 0.
			\end{equation} 
			A basic result of conditional probability states that $p(y|z) =  \sum_{x \in X} p(x|z) p(y| x,z)$. Using the concavity property of $-f(u)$ in the expression below we find 
			\begin{equation}
			\begin{split}
			- \sum_{x \in X} p(x|z) (p(y|x,z))^{r + k + 1} \ln_{\{k,r\}}(p(y|x,z)) & = - \sum_{x \in X} p(x|z) f(p(y|x,z))\\
			& \leq - f\left(\sum_{x \in X} p(x|z) p(y| x,z) \right) = - (p(y|z))^{r + k + 1} \ln_{\{k,r\}}(p(y|z)).
			\end{split}
			\end{equation}
			Multiplying both side of the above inequality with $(p(z))^{2k + 1}$ and summing over $Y$ and $Z$ we find
			\begin{equation}
			\begin{split}
			& - \sum_{y \in Y} \sum_{z \in Z} (p(z))^{2k + 1} \sum_{x \in X} p(x|z) (p(y|x,z))^{r + k + 1} \ln_{\{k,r\}}(p(y|x,z)) \\
			\leq & - \sum_{y \in Y} \sum_{z \in Z} (p(z))^{2k + 1} (p(y|z))^{r + k + 1} \ln_{\{k,r\}}(p(y|z)) = S_{\{k,r\}}(Y|Z).
			\end{split}
			\end{equation}
			Note that, $p(x,z)^{2k + 1} = (p(z))^{2k + 1} (p(x|z))^{2k + 1} \leq (p(z))^{2k + 1} p(x|z)$. Therefore,
			\begin{equation}
			\begin{split}
			S_{\{k,r\}}(Y|X,Z) & = - \sum_{x \in X} \sum_{y \in Y} \sum_{z \in Z} (p(x,z))^{2k + 1} (p(y|x,z))^{r + k + 1} \ln_{\{k,r\}}(p(y|x,z)) \\
			& \leq - \sum_{x \in X} \sum_{y \in Y} (p(z))^{2k + 1} \sum_{x \in X} p(x|z) (p(y|x,z))^{r + k + 1} \ln_{\{k,r\}}(p(y|x,z)). \\
			\end{split}
			\end{equation}
			Combining we get $S_{\{k,r\}}(Y|Z) \geq S_{\{k,r\}}(Y|X,Z)$.
		\end{proof}
		
		The above inequality leads us to the strong sub-additivity property of the generalized entropy which is mentioned below.
		
		\begin{theorem}\label{strong_sub_additive_property}
			Given any three random variable $X, Y$ and $Z$ we have
			$$S_{\{k,r\}}(X,Y, Z) + S_{\{k,r\}}(Z) \leq S_{\{k,r\}}(X, Z) + S_{\{k,r\}}(Y,Z).$$
		\end{theorem}
		
		\begin{proof}
			Theorem \ref{chain_rule_for_Sharma_MIttal_entropy} indicates
			\begin{equation}
			\begin{split}
			S_{\{k,r\}}(X, Z) + S_{\{k,r\}}(Y, Z) = & S_{\{k,r\}}(Z) + S_{\{k,r\}}(X|Z) + S_{\{k,r\}}(Z) + S_{\{k,r\}}(Y|Z)\\
			= & 2 S_{\{k,r\}}(Z) + S_{\{k,r\}}(X|Z) + S_{\{k,r\}}(Y|Z).
			\end{split}
			\end{equation}
			Now, applying the chain rules mentioned in Corollary \ref{chani_rule_2} we find
			\begin{equation}
			S_{\{k,r\}}(X, Z) + S_{\{k,r\}}(Y, Z) = 2 S_{\{k,r\}}(Z) + S_{\{k,r\}}(X,Y|Z) - S_{\{k,r\}}(Y|X,Z) + S_{\{k,r\}}(Y|Z).
			\end{equation} 
			The chain rule in Corollary \ref{chani_rule_1} leads us to
			\begin{equation}
			\begin{split}
			S_{\{k,r\}}(X, Z) + S_{\{k,r\}}(Y, Z) = & 2S_{\{k,r\}}(Z) + S_{\{k,r\}}(X, Y, Z)  - S_{\{k,r\}}(Z) - S_{\{k,r\}}(Y|X,Z) + S_{\{k,r\}}(Y|Z) \\
			= & S_{\{k,r\}}(X, Y, Z) + S_{\{k,r\}}(Z) + S_{\{k,r\}}(Y|Z) - S_{\{k,r\}}(Y|X,Z).
			\end{split}
			\end{equation} 
			Now, Lemma \ref{inequality_for_strongly_sub_additivity_property} indicates $S_{\{k,r\}}(Y|Z) - S_{\{k,r\}}(Y|X,Z) \geq 0$. Therefore, 
			\begin{equation}
			S_{\{k,r\}}(X, Z) + S_{\{k,r\}}(Y, Z) \geq S_{\{k,r\}}(X, Y, Z) + S_{\{k,r\}}(Z).
			\end{equation}
			Hence, the result follows.
		\end{proof}

	\section{Two-parameter generalized divergence}
	
		In the Shannon information theory, the relative entropy, or the Kullback-Leibler (KL) divergence is a measure of difference between two probability distributions. Recall that given two probability distributions $\mathcal{P} = \{p(x)\}_{x \in X}$ and $\mathcal{Q} = \{q(x)\}_{x \in X}$ the Kullback-Leibler divergence \cite{cover2012elements} is defined by
		\begin{equation}
		D(\mathcal{P}||\mathcal{Q}) = \sum_{x \in X} p(x) \ln\left(\frac{p(x)}{q(x)}\right) = - \sum_{x \in X} p(x) \ln\left(\frac{q(x)}{p(x)}\right). 
		\end{equation} 
		We generalize it in terms of the generalized entropy as follows:
		\begin{definition}\label{Sharma_Mittal_relative_entropy}
			(Generalized divergence) Given two probability distributions $\mathcal{P} = \{p(x)\}_{x \in X}$ and $\mathcal{Q} = \{q(x)\}_{x \in X}$ the generalized divergence is represented by
			\begin{equation*}
			D_{\{k,r\}}(\mathcal{P} || \mathcal{Q}) = \sum_{x \in X} p(x) \left(\frac{p(x)}{q(x)} \right)^{r - k} \ln_{\{k,r\}} \left(\frac{p(x)}{q(x)} \right) = - \sum_{x \in X} p(x) \left(\frac{q(x)}{p(x)} \right)^{r + k} \ln_{\{k,r\}} \left(\frac{q(x)}{p(x)} \right),
			\end{equation*}
			where $0 < k \leq \frac{1}{2}$ and $r > 0$.
		\end{definition}
		The equivalence between two expressions of $D_{\{k,r\}}(\mathcal{P} || \mathcal{Q})$ follows from equation (\ref{inversion_under_logarithm}). Putting $k = r = \frac{1 - q}{2}$ in $- \sum_{x \in X} p(x) \left(\frac{q(x)}{p(x)} \right)^{r + k} \ln_{\{k,r\}} \left(\frac{q(x)}{p(x)} \right)$ we find 	
		\begin{equation}
		D_{\left\{\frac{1 - q}{2}, \frac{1 - q}{2} \right\}} = -\sum_{x \in X} p(x) \frac{\left(\frac{q(x)}{p(x)}\right)^{1 - q} - 1}{1 - q} = D_q(\mathcal{P}||\mathcal{Q}),
		\end{equation}
		which is the Tsallis divergence \cite{furuichi2006information}, \cite{furuichi2004fundamental}. Below we discuss a few properties of the generalized divergence.
		
		\begin{lemma}
			(Non-negativity) For any two probability distribution $\mathcal{P}$ and $\mathcal{Q}$ the generalized divergence $D_{\{k,r\}}(\mathcal{P}||\mathcal{Q}) \geq 0$. The equality holds for $\mathcal{P} = \mathcal{Q}$. 
		\end{lemma} 
		\begin{proof}
			It can be proved that the function $-u^{k + r}\ln_{\{k,r\}}(u)$ is a convex function for $u \geq 0$, $0 \leq k \leq \frac{1}{2}$ and $r > 0$. Therefore, 
			\begin{equation}
			\begin{split} 
			D_{\{k,r\}}(\mathcal{P}||\mathcal{Q}) & = - \sum_{x \in X} p(x) \left(\frac{q(x)}{p(x)}\right)^{r + k} \ln_{\{k,r\}}\left(\frac{q(x)}{p(x)}\right) \\
			& \geq - \left[\sum_{x \in X} p(x) \left(\frac{q(x)}{p(x)}\right)^{r + k}\right] \ln_{\{k,r\}}\left(\sum_{x \in X} p(x) \frac{q(x)}{p(x)}\right).
			\end{split} 
			\end{equation}
			Now, $\ln_{\{k,r\}}\left(\sum_{x \in X} p(x) \frac{q(x)}{p(x)}\right) = \ln_{\{k,r\}}\left(\sum_{x \in X} q(x) \right) = \ln_{\{k,r\}}(1) = 0$. Note that, if $\mathcal{P} = \mathcal{Q}$ then 
			\begin{equation}
			\begin{split}
			D_{\{k,r\}}(\mathcal{P} || \mathcal{P}) = - \sum_{x \in X} p(x) \left(\frac{p(x)}{p(x)} \right)^{r + k} \ln_{\{k,r\}} \left(\frac{p(x)}{p(x)} \right) = - \sum_{x \in X} p(x) \ln_{\{k,r\}}(1) = 0.
			\end{split}
			\end{equation}
		\end{proof}
		
		\begin{lemma}
			(Symmetry) Let $\mathcal{P}' = \{p'_i\}$ and $\mathcal{Q}' = \{q'_i\}$ be two probability distributions, such that, $p(x)' = p_{\pi(i)}$ and $q(x)' = q_{\pi(i)}$ for a permutation $\pi$ and probability distributions $\mathcal{P} = \{p(x)\}_{x \in X}$ and $\mathcal{Q} = \{q(x)\}_{x \in X}$. Then $D_{\{k,r\}}(\mathcal{P}'||\mathcal{Q}') = D_{\{k,r\}}(\mathcal{P}||\mathcal{Q})$.
		\end{lemma}
		\begin{proof}
			The permutation $\pi$ alters the position of $p(x) \left(\frac{p(x)}{q(x)}\right)^{r - k} \ln_{\{k,r\}} \left(\frac{p(x)}{q(x)}\right)$ under addition and keeps the sum $D_{\{k,r\}}(\mathcal{P}||\mathcal{Q})$, unaltered. Hence, the proof follows trivially.
		\end{proof}
		
		\begin{lemma}
			(Possibility of extension) Let $\mathcal{P}' = \mathcal{P} \cup \{0\}$ and $\mathcal{Q}' = \mathcal{Q} \cup \{0\}$, then $D_{\{k,r\}}(\mathcal{P}'||\mathcal{Q}') = D_{\{k,r\}}(\mathcal{P}||\mathcal{Q})$.
		\end{lemma}
		\begin{proof}
			Define $0\left(\frac{0}{0}\right)^{r + k}\ln_{\{k,r\}}\left(\frac{0}{0}\right) = \lim_{(x,y) \rightarrow (0,0)} x \left(\frac{y}{x}\right)^{r + k} \ln_{\{k,r\}} \left(\frac{y}{x}\right)$. Note that, 
			$$\lim\limits_{x \rightarrow 0} \lim\limits_{y \rightarrow 0} x \left(\frac{y}{x}\right)^{r + k} \ln_{\{k,r\}} \left(\frac{y}{x}\right) = 0.$$ 
			In addition, we can write that $\lim_{y \rightarrow 0} \lim_{x \rightarrow 0} x \left(\frac{y}{x}\right)^{r + k} \ln_{\{k,r\}} \left(\frac{y}{x}\right) = 0$. Now applying Moore-Osgood Theorem \cite{stewart1995multivariable} we find that $\lim_{(x, y) \rightarrow (0, 0)} x \left(\frac{y}{x}\right)^{r + k} \ln_{\{k,r\}} \left(\frac{y}{x}\right) = 0$. Therefore, $0\ln_{\{k,r\}}\left(\frac{0}{0}\right) = 0$. Hence, $D_{\{k,r\}}(\mathcal{P}'||\mathcal{Q}') = D_{\{k,r\}}(\mathcal{P}||\mathcal{Q})$.
		\end{proof}
		
		Given two probability distributions $\mathcal{P} = \{p(x)\}_{x \in X}$ and $\mathcal{Q} = \{q(y)\}_{y \in Y}$ we can define a joint probability distribution $\mathcal{P} \otimes \mathcal{Q} = \{p(x)q(y)\}_{(x,y) \in X \otimes Y}$. Note that, for all $x \in X$ and $y \in Y$ we have $0 \leq p(x)q(y) \le 1$. In addition, $\sum_{x \in X} \sum_{y \in Y} p(x) q(y) = 1$. Now, we have the following theorem.
		
		\begin{theorem}\label{pseudo_additivity_of_divergence}
			(Pseudo-additivity) Given probability distributions $\mathcal{P}^{(1)} = \{p^{(1)}(x)\}_{x \in X}$, $\mathcal{Q}^{(1)} = \{q^{(1)}(x)\}_{x \in X}$, $\mathcal{P}^{(2)} = \{p^{(2)}(y)\}_{y \in Y}$ and $\mathcal{Q}^{(2)} = \{q^{(2)}(y)\}_{y \in Y}$ we have 
			\begin{equation*}
			\begin{split}	
			D_{\{k,r\}}(\mathcal{P}^{(1)} \otimes \mathcal{P}^{(2)} || \mathcal{Q}^{(1)} \otimes \mathcal{Q}^{(2)}) =  D_{\{k,r\}}(\mathcal{P}^{(1)} || \mathcal{Q}^{(1)}) + D_{\{k,r\}}(\mathcal{P}^{(2)} || \mathcal{Q}^{(2)}) - 2k D_{\{k,r\}}(\mathcal{P}^{(1)} || \mathcal{Q}^{(1)}) D_{\{k,r\}}(\mathcal{P}^{(2)} || \mathcal{Q}^{(2)}).			
			\end{split}
			\end{equation*}
		\end{theorem}
		
		\begin{proof}
			Recall the product rule of $\ln_{\{k,r\}}(xy)$ mentioned in Lemma \ref{product_1}. Expanding the logarithm we find 
			\begin{equation}
			\begin{split}
			& \left(\frac{q^{(1)}(x) q^{(2)}(y)}{p^{(1)}(x) p^{(2)}(y)}\right)^{r + k} \ln_{\{k,r\}} \left(\frac{q^{(1)}(x) q^{(2)}(y)}{p^{(1)}(x) p^{(2)}(y)}\right)\\
			= &\left(\frac{q^{(1)}(x)}{p^{(1)}(x)}\right)^{r + k} \ln_{\{k,r\}} \left(\frac{q^{(1)}(x) }{p^{(1)}(x) }\right) + \left(\frac{q^{(2)}(y)}{p^{(2)}(y)}\right)^{r + k} \ln_{\{k,r\}} \left(\frac{q^{(2)}(y) }{p^{(2)}(y) }\right)\\
			& + 2k \left(\frac{q^{(1)}(x)}{p^{(1)}(x)}\right)^{r + k} \ln_{\{k,r\}} \left(\frac{q^{(1)}(x) }{p^{(1)}(x) }\right)  \left(\frac{q^{(2)}(y)}{p^{(2)}(y)}\right)^{r + k} \ln_{\{k,r\}} \left(\frac{q^{(2)}(y) }{p^{(2)}(y) }\right).
			\end{split}
			\end{equation}
			Multiplying $p^{(1)}(x) p^{(2)}(y)$ with both side we find
			\begin{equation} 
			\begin{split}
			& - p^{(1)}(x) p^{(2)}(y) \left(\frac{q^{(1)}(x) q^{(2)}(y)}{p^{(1)}(x) p^{(2)}(y)}\right)^{r + k} \ln_{\{k,r\}} \left(\frac{q^{(1)}(x) q^{(2)}(y)}{p^{(1)}(x) p^{(2)}(y)}\right)\\
			= & - p^{(1)}(x) \left(\frac{q^{(1)}(x)}{p^{(1)}(x)}\right)^{r + k} \ln_{\{k,r\}} \left(\frac{q^{(1)}(x) }{p^{(1)}(x) }\right) p^{(2)}(y) \\
			& - p^{(2)}(y) \left(\frac{q^{(2)}(y)}{p^{(2)}(y)}\right)^{r + k} \ln_{\{k,r\}} \left(\frac{q^{(2)}(y) }{p^{(2)}(y) }\right) p^{(1)}(x)\\
			& - 2k \times p^{(1)}(x)\left(\frac{q^{(1)}(x)}{p^{(1)}(x)}\right)^{r + k} \ln_{\{k,r\}} \left(\frac{q^{(1)}(x) }{p^{(1)}(x) }\right) \\
			& \times p^{(2)}(y) \left(\frac{q^{(2)}(y)}{p^{(2)}(y)}\right)^{r + k} \ln_{\{k,r\}} \left(\frac{q^{(2)}(y) }{p^{(2)}(y) }\right).
			\end{split}
			\end{equation}
			Now, applying Definition \ref{Sharma_Mittal_relative_entropy} we find $D_{\{k,r\}}(\mathcal{P}^{(1)} \otimes \mathcal{P}^{(2)} || \mathcal{Q}^{(1)} \otimes \mathcal{Q}^{(2)})$
			\begin{equation}
			\begin{split}
			= & - \left[\sum_{x \in X} p^{(1)}(x) \left(\frac{q^{(1)}(x)}{p^{(1)}(x)}\right)^{r + k} \ln_{\{k,r\}} \left(\frac{q^{(1)}(x) }{p^{(1)}(x) }\right) \right] \left[ \sum_{y \in Y} p^{(2)}(y) \right] \\
			& - \left[ \sum_{y \in Y} p^{(2)}(y) \left(\frac{q^{(2)}(y)}{p^{(2)}(y)}\right)^{r + k} \ln_{\{k,r\}} \left(\frac{q^{(2)}(y) }{p^{(2)}(y) }\right) \right] \left[\sum_{x \in X} p^{(1)}(x) \right] \\
			& - 2k \times \left[ \sum_{x \in X} p^{(1)}(x) \left(\frac{q^{(1)}(x)}{p^{(1)}(x)}\right)^{r + k} \ln_{\{k,r\}} \left(\frac{q^{(1)}(x) }{p^{(1)}(x) }\right) \right] \\
			& \times \left[ \sum_{y \in Y} p^{(2)}(y) \left(\frac{q^{(2)}(y)}{p^{(2)}(y)}\right)^{r + k} \ln_{\{k,r\}} \left(\frac{q^{(2)}(y) }{p^{(2)}(y) }\right) \right] \\
			= & D_{\{k,r\}}(\mathcal{P}^{(1)} || \mathcal{Q}^{(1)}) + D_{\{k,r\}}(\mathcal{P}^{(2)} || \mathcal{Q}^{(2)}) - 2k D_{\{k,r\}}(\mathcal{P}^{(1)} || \mathcal{Q}^{(1)}) D_{\{k,r\}}(\mathcal{P}^{(2)} || \mathcal{Q}^{(2)}).
			\end{split}
			\end{equation}
		\end{proof}
		
		The next theorem needs the log-sum inequality for $\ln_{\{k , r\}}$, which we mention in the next lemma.
		\begin{lemma}\label{log_sum_inequality}
			Let $a_1, a_2, \dots a_n$ and $b_1, b_2, \dots b_n$ be non-negative numbers. In addition, $a = \sum_{i = 1}^n a_i$ and $b = \sum_{i = 1}^n b_i$. Then,
			$$\sum_{i = 1}^n a_i \left(\frac{a_i}{b_i}\right)^{r - k}\ln_{\{k,r\}} \left(\frac{a_i}{b_i}\right) \geq a \left(\frac{a}{b}\right)^{r - k} \ln_{\{k,r\}} \left(\frac{a}{b}\right).$$
		\end{lemma}
		\begin{proof}
			\begin{equation}
			\begin{split}
			\sum_{i = 1}^n a_i \left(\frac{a_i}{b_i}\right)^{r - k}\ln_{\{k,r\}} \left(\frac{a_i}{b_i}\right) & = b \sum_{i = 1}^n \frac{b_i}{b} \frac{a_i}{b_i} \left(\frac{a_i}{b_i}\right)^{r - k} \ln_{\{k,r\}} \left(\frac{a_i}{b_i}\right) = b \sum_{i = 1}^n \frac{b_i}{b} f\left(\frac{a_i}{b_i}\right).
			\end{split}
			\end{equation}
			We can prove that the function $f(x) = x^{r - k + 1} \ln_{\{k,r\}}(x)$ is a convex function $x > 0$ and for $0 < k \leq \frac{1}{2}$. Therefore, 
			\begin{equation}
			\begin{split}
			\sum_{i = 1}^n a_i \left(\frac{a_i}{b_i}\right)^{r - k} \ln_{\{k,r\}} \left(\frac{a_i}{b_i}\right) & \geq b f\left(\sum_{i = 1}^n \frac{b_i}{b} \frac{a_i}{b_i}\right) = b f\left(\frac{1}{b} \sum_{i = 1}^n a_i\right) = b f\left(\frac{a}{b}\right) = b \left(\frac{a}{b}\right)^{r - k + 1}\ln_{\{k,r\}}\left(\frac{a}{b}\right),
			\end{split}
			\end{equation}
			which indicates the proof.
		\end{proof}
		
		\begin{theorem}\label{joint_convexity_of_divergence}
			(Joint convexity) Let $\mathcal{P}^{(k)} = \{p^{(k)}(x)\}_{x \in X}$ and $\mathcal{Q}^{(k)} = \{q^{(k)}(x)\}_{x \in X}$ for $k = 1, 2$ are probability distributions. Construct new probability distributions $(1 - \lambda)\mathcal{P}^{(1)} + \lambda \mathcal{P}^{(2)} = \{(1 - \lambda)p^{(1)}(x) + \lambda p^{(2)}(x)\}_{x \in X}$, and $(1 - \lambda)\mathcal{Q}^{(1)} + \lambda \mathcal{Q}^{(2)} = \{(1 - \lambda)q^{(1)}(x) + \lambda q^{(2)}(x)\}_{x \in X}$ as convex combinations. Then,
			\begin{equation*} 
			\begin{split} 
			D_{\{k, r\}}((1 - \lambda)\mathcal{P}^{(1)} & + \lambda \mathcal{P}^{(2)} || (1 - \lambda) \mathcal{Q}^{(1)} + \lambda \mathcal{Q}^{(2)}) \leq (1 - \lambda) D_{\{k, r\}}(\mathcal{P}^{(1)} || \mathcal{Q}^{(1)}) + \lambda D_{\{k, r\}}(\mathcal{P}^{(2)} || \mathcal{Q}^{(2)}).
			\end{split} 
			\end{equation*} 
		\end{theorem}
		
		\begin{proof}
			Note that, $D_{\{k, r\}}((1 - \lambda)\mathcal{P}^{(1)} + \lambda \mathcal{P}^{(2)} || (1 - \lambda) \mathcal{Q}^{(1)} + \lambda \mathcal{Q}^{(2)}) = $
			\begin{equation}
			\sum_{x \in X} ((1 - \lambda)p^{(1)}(x) + \lambda p^{(2)}(x)) \left(\frac{(1 - \lambda) p^{(1)}(x) + \lambda p^{(2)}(x)}{(1 - \lambda) q^{(1)}(x) + \lambda q^{(2)}(x)}\right)^{r - k} \ln_{\{k,r\}} \left(\frac{(1 - \lambda)p^{(1)}(x) + \lambda p^{(2)}(x)}{(1 - \lambda) q^{(1)}(x) + \lambda q^{(2)}(x)}\right).
			\end{equation}
			Now, applying the log-sum inequality stated in Lemma \ref{log_sum_inequality} we find
			\begin{equation}
			\begin{split}
			& ((1 - \lambda) p^{(1)}(x) + \lambda p^{(2)}(x)) \left(\frac{(1 - \lambda) p^{(1)}(x) + \lambda p^{(2)}(x)}{(1 - \lambda) q^{(1)}(x) + \lambda q^{(2)}(x)}\right)^{r - k} \ln_{\{k,r\}} \left(\frac{(1 - \lambda) p^{(1)}(x) + \lambda p^{(2)}(x)}{(1 - \lambda) q^{(1)}(x) + \lambda q^{(2)}(x)}\right)\\
			\leq & (1 - \lambda) p^{(1)}(x) \left(\frac{(1 - \lambda) p^{(1)}(x)}{(1 - \lambda) q^{(1)}(x)}\right)^{r - k} \ln_{\{k,r\}} \left(\frac{(1 - \lambda) p^{(1)}(x)}{(1 - \lambda) q^{(1)}(x)}\right) + \lambda p^{(2)}(x) \left(\lambda \frac{p^{(2)}(x)}{\lambda q^{(2)}(x)}\right)^{r - k} \ln_{\{k,r\}} \left(\frac{\lambda p^{(2)}(x)}{\lambda q^{(2)}(x)}\right).
			\end{split}
			\end{equation}
			Summing over $x$, we find the result.
		\end{proof}
		
		Consider a transition probability matrix $W = (w_{j,i})_{m \times n}$, such that, $\sum_{j = 1}^m w_{j,i} = 1$ for all $i = 1, 2, \dots n$. Let $\mathcal{P} = \{p_i^{(in)}\}_{i = 1}^n$ and $\mathcal{Q} = \{q_i^{(in)}\}_{i = 1}^n$ be two probability distributions. After a transition with $W$ the new probability distributions are $W\mathcal{P} = \{p_j^{(out)}\}_{j = 1}^m$ and $W\mathcal{Q} = \{q_j^{(out)}\}_{j = 1}^m$, respectively, where $p_j^{(out)} = \sum_{i = 1}^n w_{j,i} p_i^{(in)}$, and $q_j^{(out)} = \sum_{i = 1}^n w_{j,i} q_i^{(in)}$. Now, we have the following theorem.
		
		\begin{theorem}\label{information_monotonicity}
			(Information monotonicity) Given probability distributions $\mathcal{P}$, $\mathcal{Q}$ and transition probability matrix $W$ we have $D_{\{k,r\}}(W\mathcal{P}|| W\mathcal{Q}) \leq D_{\{k,r\}}(\mathcal{P}|| \mathcal{Q})$.
		\end{theorem}
		
		\begin{proof}
			Definition \ref{Sharma_Mittal_relative_entropy} of the generalized divergence indicates that
			\begin{equation}
			\begin{split}
			D_{\{k,r\}}(W\mathcal{P} || W\mathcal{Q}) = & \sum_{j = 1}^m p_j^{(out)} \left(\frac{p_j^{(out)}}{q_j^{(out)}} \right)^{r - k} \ln_{\{k,r\}} \left(\frac{p_j^{(out)}}{q_j^{(out)}} \right) \\
			= & \sum_{j = 1}^m \left[\sum_{i = 1}^n w_{ji} p_i^{(in)}\right] \left( \frac{\sum_{i = 1}^n w_{ji} p_i^{(in)}}{\sum_{i = 1}^n w_{ji} q_i^{(in)}} \right)^{r - k} \ln_{\{k,r\}}\left( \frac{\sum_{i = 1}^n w_{ji} p_i^{(in)}}{\sum_{i = 1}^n w_{ji} q_i^{(in)}} \right). 
			\end{split}
			\end{equation}
			Now, from Lemma \ref{log_sum_inequality} we find that
			\begin{equation}
			\begin{split}
			D_{\{k,r\}}(W\mathcal{P} || W\mathcal{Q}) \leq & \sum_{j = 1}^m \sum_{i = 1}^n \left(w_{ji} p_i^{(in)}\right) \left( \frac{w_{ji} p_i^{(in)}}{w_{ji} q_i^{(in)}} \right)^{r - k} \ln_{\{k,r\}}\left( \frac{w_{ji} p_i^{(in)}}{w_{ji} q_i^{(in)}} \right) \\
			= & \sum_{i = 1}^n \left[ p_i^{(in)} \left( \frac{p_i^{(in)}}{q_i^{(in)}} \right)^{r - k} \ln_{\{k,r\}}\left( \frac{p_i^{(in)}}{q_i^{(in)}} \right) \right] \left[ \sum_{j = 1}^m w_{ji} \right]\\
			= & \sum_{j = 1}^m p_i^{(in)} \left( \frac{p_i^{(in)}}{q_i^{(in)}} \right)^{r - k} \ln_{\{k,r\}}\left( \frac{p_i^{(in)}}{q_i^{(in)}} \right) ~\text{since}~ \sum_{j = 1}^m w_{ji} = 1.
			\end{split}
			\end{equation}
			Hence, we have $D_{\{k,r\}}(W\mathcal{P} || W\mathcal{Q}) \leq D_{\{k,r\}}(\mathcal{P} || \mathcal{Q})$.
		\end{proof}
		
		In Theorem \ref{information_monotonicity}, if the probability transition matrix $W = (w_{ji})_{m \times n}$ has $m < n$, then $W$ partitions the random variable $X = (x_1, x_2, \dots x_n)$ into $m$ groups $G_1, G_2, \dots G_n$ such that $X = \cup_{j = 1}^m G_j$, and $G_k \cap G_l = \emptyset$. Then $p_j^{(out)}(G_j) = \sum_{x_i \in G_j} p_i^{(in)}$. Now Theorem \ref{information_monotonicity} indicates $D(W\mathcal{P}|| W\mathcal{Q}) \leq D(\mathcal{P}|| \mathcal{Q})$, which is formally mentioned as information monotonicity.

	\section{Information geometric aspects}
	
		This section is dedicated to the geometric nature of the generalized divergence. First recall a number of fundamental concepts of information geometry \cite{amari2007methods}. A probability simplex is given by, 
		\begin{equation}
		S = \{\mathcal{P}: \mathcal{P} = (p_1, p_2, \dots p_n), 0 \leq p_i \leq 1, \sum_{i = 1}^n p_i = 1\}.
		\end{equation}
		with the distribution $\mathcal{P}$ described by $n$-independent probabilities $(p_1, p_2, \dots p_n)$. Consider a parametric family of distributions $\mathcal{P}({\bf x})$ with parameter vector ${\bf x} = (x_1, x_2, \dots x_n) \in X$, where $X$ is a parameter space. If the parameter space $X$ is a differentiable manifold and the mapping $x \mapsto \mathcal{P}({\bf p},{\bf x})$ is a diffeomorphism we can identify statistical models in the family as points on the manifold $X$. The Fisher-Rao information matrix $E(ss^T)$, where $s$ is the gradient $[s]_i = \frac{\partial \log \mathcal{P}({\bf p},{\bf x})}{\partial x_i}$ may be used to endow $X$ with the following Riemannian metric
		\begin{equation} 
			G_x(u,v) = \sum_{i,j}u_iv_j \int \mathcal{P}({\bf p},{\bf x}) \frac{\partial}{\partial x_i}\log\mathcal{P}({\bf p},{\bf x}) \frac{\partial}{\partial x_j}\log\mathcal{P}({\bf p},{\bf x})dp = \sum_{i,j}u_iv_j E \left(\frac{\partial\log\mathcal{P}({\bf p},{\bf x})}{\partial x_i} \frac{\partial\log\mathcal{P}({\bf p},{\bf x})}{\partial x_j} \right).
		\end{equation} 
		If $X$ is a discrete random variable then the above integral is replaced with a sum. An equivalent form of $G_x(u, v)$ for normalized distributions is given by
		\begin{equation}
			G_x(u,v) =  -\sum_{i,j}u_iv_j\int \mathcal{P}({\bf p},{\bf x}) \frac{\partial^2}{\partial x_j \partial x_i}\log\mathcal{P}({\bf p},{\bf x})dp = \sum_{i,j}u_iv_j E\left( - \frac{\partial^2}{\partial x_j\partial x_i} \log\mathcal{P}({\bf p},{\bf x}) \right).
		\end{equation} 
		In information geometry, a function $D(\mathcal{P}|| \mathcal{Q})$ for $\mathcal{P}, \mathcal{Q} \in S$ is called divergence if $D(\mathcal{P}|| \mathcal{Q}) \geq 0$ and $D(\mathcal{P}|| \mathcal{Q}) = 0$ if and only if $\mathcal{P} = \mathcal{Q}$. Consider a point $\mathcal{P}$ with coordinates $(p_1, p_2, \dots p_n)$. Let $\mathcal{Q}  = (\mathcal{P} + d(\mathcal{P}))$ be another point infinitesimally close to $\mathcal{P}$. Using the Taylor series expansion we have 
		\begin{equation}
		D(\mathcal{P} + d\mathcal{P} || \mathcal{P}) = \sum g_{ij} dp_i dp_j + O(|dp|^3),
		\end{equation} 
		where $g_{ij}$ is a positive-definite matrix. 	Hence, the Riemannian metric induced by the divergence $D$ is given by
		\begin{equation}
		g_{ij}(\mathcal{P}) = \frac{\partial^2}{\partial p_i \partial p_j} D_{\{k,r\}}(\mathcal{P}|| \mathcal{Q}) |_{\mathcal{Q} = \mathcal{P}}.
		\end{equation}
		Thus, the divergence gives us a means of determining the degree of separation between two
		points on a manifold. It is not a metric since it is not necessarily symmetric. Also, the length of small line segment is given by
		\begin{equation}
		ds^2 = \frac{1}{2} D(\mathcal{P}|| \mathcal{P} + d\mathcal{P}).
		\end{equation}
		
		Recalling Definition \ref{Sharma_Mittal_relative_entropy} of the generalized divergence we calculate
		\begin{equation}
		\begin{split}
		& \frac{\partial}{\partial p_i} D_{\{k,r\}}(\mathcal{P} || \mathcal{Q}) = \frac{\partial}{\partial p_i} \left[ p_i \left(\frac{p_i}{q_i} \right)^{r - k} \ln_{\{k,r\}} \left(\frac{p_i}{q_i} \right) \right]\\
		& \hspace{2.5cm} = \frac{\left((2 r+1) \left(\left(\frac{p_i}{q_i}\right){}^{2 k}-1\right)+2 k\right) \left(\frac{p_i}{q_i}\right){}^{2 r-2 k}}{2 k} \\
		& \frac{\partial^2}{\partial^2 p_i} D_{\{k,r\}}(\mathcal{P} || \mathcal{Q}) = \frac{\left(r (2 r+1) \left(\left(\frac{p_i}{q_i}\right){}^{2 k}-1\right)-2 k^2+4 k r+k\right) \left(\frac{p_i}{q_i}\right){}^{2 r-2 k}}{k p_i} \\
		& \frac{\partial^2}{\partial^2 p_i} D_{\{k,r\}}(\mathcal{P} || \mathcal{Q}) |_{\mathcal{Q} = \mathcal{P}} = \frac{-2 k + 4r + 1}{p_i}, \\
		& \frac{\partial^2}{\partial p_j \partial p_i} D_{\{k,r\}}(\mathcal{P} || \mathcal{Q}) = 0.
		\end{split}
		\end{equation}
		Therefore, the Fisher information matrix $G = (g_{ij})_{n \times n}$ for the generalized divergence is given by
		\begin{equation}
		g_{ij} = \begin{cases} \frac{-2 k + 4r + 1}{p_i}, & \text{for}~ i = j \\ 0 & \text{for}~ i \neq j. \end{cases}
		\end{equation}
		A manifold is called Hassian if there is a function $\Psi(u)$ such that $g_{ij}(\mathcal{P}) = \partial_{ij}(\Psi)$. Here, for $i = j$ we have $\partial_{ii}(\Psi) = g_{ii}(u) = \frac{1 - 2 k + 4 r}{u}$. Integrating twice we find 
		\begin{equation}
		\Psi_{ii}(u) = c_2 + u(c_1 +2k - 4r - 1) + (-2 k+4 r+1) u\log (u),
		\end{equation}
		where $c_1$ and $c_2$ are integrating constants. For $i \neq j$ we have $\partial_{ii}(\Psi) = g_{ij} = 0$, that is $\Psi(u) = c_1u + c_2$. Hence, the statistical manifold induced by the generalized divergence is Hassian.

	\section{Conclusion} 
	
		In recent years, the idea of entropy offers a broad scope of mathematical investigations. In this article, we introduce the two parameter deformed entropy $\ln_{\{k, r\}}$. Interestingly, it can be reduced to the $q$-deformed logarithm for $k = r = \frac{q - 1}{2}$ and natural logarithm when $q \rightarrow 1$. In table \ref{Comparison_between_different_logarithms}, we compare various properties of the logarithm, the $q$-deformed logarithm and $\ln_{\{k, r\}}$. It leads us to propose the new generalized entropy $S_{\{k, r\}}$ with two parameters $k$ and $r$. Interestingly, our proposed entropy has a number of important characteristics which are not established in the earlier proposals of two parameter generalized entropy. The table \ref{Comparison_between_different_entropy} contains the comparative properties of the Shannon entropy, the Tsallis entropy, and $S_{\{k, r\}}$. The table suggests that the new generalized entropy is efficient to be utilized in classical information theory. These properties include chain rule, pseudo-additive property, sub-additive property, and information monotonicity. Properties of the two parameter generalized divergence $D_{\{k, r\}}$, the Tsallis divergence, and the Kullback–Leibler divergence are collected in table \ref{Comparison_between_different_divergence}. Also, we justify that the statistical manifold induced by the generalized divergence is Hassian. 
		
		An interested reader may extend this work further. In the Shannon information theory, the mutual information of two random variables $X$ and $Y$ is defined by $I(X;Y) = D(p(x,y)| p(x)p(y))$, which is the Kullback-Leibler divergence between two probability distributions $p(x,y)$ and $p(x)p(y)$. In case of the generalized entropy, one may introduce the mutual information $I_{\{k,r\}}(X;Y) = D_{\{k,r\}}(p(x,y)||p(x)p(y))$ then investigates its properties. Moreover, the mutual information has a crucial role in the literature of data processing inequalities. Hence, two parameter deformation of data-processing inequalities will be very crucial in this direction.

		\begin{table}[h!]
			\centering
			\caption{Comparison between different logarithms}
			\label{Comparison_between_different_logarithms}
			\begin{tabular}{| p{4cm} | p{2cm} | p{10cm} |}
				\hline
				Properties with descriptions & Logarithm & Expressions\\ \hline \hline 
				\multirow{3}{=}{Definition of logarithm} & logarithm & $\log(x)$.  \\ \cline{2-3} 
				& $q$-deformed logarithm & $\ln_q(u) = \frac{u^{1 - q} - 1}{1 - q}$ for $q \neq 1$ \cite{yamano2002some} \\ \cline{2-3}
				& $\ln_{\{k , r\}}$ & $\ln_{\{k,r\}}(u) = \frac{u^{k} - u^{-k}}{2k u^{r}} = \frac{u^{2k} - 1}{2k u^{r + k}},$
				with $r > 0$ and $0 < k \leq 1$. (Definition \ref{redefined_lnkr})\\ \hline \hline 
				\multirow{3}{=}{Product law: Let $u$ and $v$ be two non-zero real numbers, then } & logarithm & $\log(uv) = \log(u) + \log(v)$ \\ \cline{2-3} 
				& $q$-deformed logarithm & $\ln_q(uv) = \ln_q(u) + \ln_q(v) + (1 - q) \ln_q(u) \ln_q(v)$ \cite{yamano2002some} \\ \cline{2-3}
				& $\ln_{\{k , r\}}$ & $(uv)^{r + k} \ln_{\{k,r\}}(uv) = u^{r + k} \ln_{\{k,r\}}(u) + v^{r + k} \ln_{\{k,r\}}(v) + 2k u^{r + k} v^{r + k} \ln_{\{k,r\}}(u)\ln_{\{k,r\}}(v)$ (Lemma \ref{product_1})\\ \hline \hline 
				\multirow{3}{=}{Log sum inequality: Let $a_1, a_2, \dots a_n$ and $b_1, b_2, \dots b_n$ be non-negative numbers. In addition, $a = \sum_{i = 1}^n a_i$ and $b = \sum_{i = 1}^n b_i$. Then,} & logarithm & $\sum _{i=1}^{n}a_{i}\log {\frac {a_{i}}{b_{i}}} \geq a\log {\frac {a}{b}}$ \\ \cline{2-3} 
				& $q$-deformed logarithm & $\sum_{i = 1}^n a_i \ln_q \left(\frac{a_i}{b_i} \right) \geq a \ln_q \left(\frac{a}{b}\right)$ \cite{furuichi2004fundamental} \\ \cline{2-3}
				& $\ln_{\{k , r\}}$ & $\sum_{i = 1}^n a_i \left(\frac{a_i}{b_i}\right)^{r - k}\ln_{\{k,r\}} \left(\frac{a_i}{b_i}\right) \geq a \left(\frac{a}{b}\right)^{r - k} \ln_{\{k,r\}} \left(\frac{a}{b}\right)$ (Lemma \ref{log_sum_inequality}) \vspace{.5 cm}\\ \hline  
			\end{tabular}
		\end{table}
		
		\begin{table}
			\centering
			\caption{Comparison between different entropy}
			\label{Comparison_between_different_entropy}
			\begin{tabular}{| p{4.5cm} | p{1.5cm} | p{10cm} |}
				\hline
				Properties with descriptions & Entropy & Expressions\\ \hline \hline 
				\multirow{3}{=}{Definition of entropy: Given a random variable $X$ with probability distribution $\mathcal{P} = \{p(x)\}_{x \in X}$} & Shannon entropy & $H(X) = - \sum_{x \in X} p(x) \log(p(x)) = \sum_{x \in X} p(x) \log \left(\frac{1}{p(x)}\right)$ \\ \cline{2-3} 
				& Tsallis entropy & $S_q(X) = - \sum_{x \in X} (p(x))^q \ln_q(p(x))$ \\ \cline{2-3}
				& $S_{\{k , r\}}$ & $S_{\{k,r\}}(X) = - \sum_{x \in X} \left(p(x)\right)^{r + k + 1} \ln_{\{k,r\}}(p(x)) = \sum_{x \in X} \left(p(x)\right)^{k - r + 1} \ln_{\{k,r\}} \left( \frac{1}{p(x)} \right)$ (Definition \ref{modified_definition_of_SM_entropy})\\ \hline \hline 
				\multirow{3}{=}{Positivity } & Shannon entropy & $H(X) \geq 0$ \\ \cline{2-3} 
				& Tsallis entropy & $S_q(X) \geq 0$ \\ \cline{2-3}
				& $S_{\{k , r\}}$ & $S_{\{k,r\}}(X) \geq 0$ \\ \hline \hline 
				\multirow{3}{=}{Chain rule for independent random variables $X$ and $Y$ } & Shannon entropy & $H(X, Y) = H(X) + H(Y)$ \\ \cline{2-3} 
				& Tsallis entropy & $S_q(X, Y) = S_q(X) + S_q(Y) + (1 - q) S_q(X) S_q(Y)$ \cite{furuichi2006information}\\ \cline{2-3}
				& $S_{\{k , r\}}$ & $S_{\{k,r\}}(X, Y) = S_{\{k,r\}}(X) + S_{\{k,r\}}(Y) - 2k S_{\{k,r\}}(X) S_{\{k,r\}}(Y)$ (Equation \ref{pseudo_additive_property_of_SM})\\ \hline  \hline 
				\multirow{3}{=}{Chain rule for dependent random variables $X$ and $Y$} & Shannon entropy & $H(X, Y) = H(X) + H(Y|X)$ \\ \cline{2-3} 
				& Tsallis entropy & $S_q(X, Y) = S_q(X) + S_q(Y|X).$ \cite{furuichi2006information} \\ \cline{2-3}
				& $S_{\{k , r\}}$ & $S_{\{k,r\}}(X, Y) = S_{\{k,r\}}(X) + S_{\{k,r\}}(Y|X).$ (Theorem \ref{chain_rule_for_Sharma_MIttal_entropy})\\ \hline \hline 
				\multirow{3}{=}{Sub-additive property: Given random variables $X_1, X_2, \dots X_n$,} & Shannon entropy & $H(X_1, X_2, \dots X_n) \leq \sum_{i = 1}^n H(X_i)$ \\ \cline{2-3} 
				& Tsallis entropy & $S_q(X_1, X_2, \dots X_n) \leq \sum_{i = 1}^n S_q(X_i)$ \cite{furuichi2006information} \\ \cline{2-3}
				& $S_{\{k , r\}}$ &  $S_{\{k,r\}}(X_1, X_2, \dots X_n) \leq \sum_{i = 1}^n S_{\{k,r\}}(X_i)$ (Theorem \ref{sub_additive_property}) \\ \hline \hline 
				\multirow{3}{=}{Strong sub-additive property: Given any three random variable $X, Y$ and $Z$ we have } & Shannon entropy & $H(X,Y, Z) + H(Z) \leq H(X, Z) + H(Y,Z)$. \\ \cline{2-3} 
				& Tsallis entropy & $S_q(X, Y, Z) + S_q(Z) \leq S_q(X, Z) + S_q(Y, Z)$ \cite{furuichi2006information} \\ \cline{2-3}
				& $S_{\{k , r\}}$ & $S_{\{k,r\}}(X,Y, Z) + S_{\{k,r\}}(Z) \leq S_{\{k,r\}}(X, Z) + S_{\{k,r\}}(Y,Z)$. (Theorem \ref{strong_sub_additive_property})\\ \hline
			\end{tabular}
		\end{table}
		
		\begin{table}
			\centering
			\caption{Comparison between different divergence}
			\label{Comparison_between_different_divergence}
			\begin{tabular}{| p{4cm} | p{2cm} | p{10cm} |}
				\hline
				Properties with descriptions & Divergence & Expressions \\ \hline
				\multirow{3}{=}{Definition of divergence: Given two probability distributions $\mathcal{P} = \{p(x)\}_{x \in X}$ and $\mathcal{Q} = \{q(x)\}_{x \in X}$} &  KL divergence & $D(\mathcal{P}||\mathcal{Q}) = \sum_{x \in X} p(x) \ln\left(\frac{p(x)}{q(x)}\right) = - \sum_{x \in X} p(x) \ln\left(\frac{q(x)}{p(x)}\right)$.\\ \cline{2-3} 
				& Tsallis divergence & $D_q (\mathcal{P}||\mathcal{Q}) = - \sum_{x \in X} p(x) \ln_q \left(\frac{q(x)}{p(x)} \right)$ \cite{furuichi2004fundamental} \\ \cline{2-3}
				& $D_{\{k , r\}}$ & $D_{\{k,r\}}(\mathcal{P} || \mathcal{Q}) = \sum_{x \in X} p(x) \left(\frac{p(x)}{q(x)} \right)^{r - k} \ln_{\{k,r\}} \left(\frac{p(x)}{q(x)} \right) = - \sum_{x \in X} p(x) \left(\frac{q(x)}{p(x)} \right)^{r + k} \ln_{\{k,r\}} \left(\frac{q(x)}{p(x)} \right)$ (Definition \ref{Sharma_Mittal_relative_entropy})\\ \hline \hline 
				\multirow{3}{=}{Non-negativity} & KL divergence & $D(\mathcal{P} || \mathcal{Q}) \geq 0$ \\ \cline{2-3} 
				& Tsallis divergence & $D_q (\mathcal{P}||\mathcal{Q}) \geq 0$ \\ \cline{2-3}
				& $D_{\{k , r\}}$ & $D_{\{k , r\}}(\mathcal{P} || \mathcal{Q}) \geq 0$ \\ \hline \hline 
				\multirow{3}{=}{Pseudo-additivity: Given probability distributions $\mathcal{P}^{(1)} = \{p^{(1)}(x)\}_{x \in X}$, $\mathcal{Q}^{(1)} = \{q^{(1)}(x)\}_{x \in X}$, $\mathcal{P}^{(2)} = \{p^{(2)}(y)\}_{y \in Y}$ and $\mathcal{Q}^{(2)} = \{q^{(2)}(y)\}_{y \in Y}$ we have } & KL divergence & $D(\mathcal{P}^{(1)} \otimes \mathcal{P}^{(2)} || \mathcal{Q}^{(1)} \otimes \mathcal{Q}^{(2)}) = D(\mathcal{P}^{(1)} || \mathcal{Q}^{(1)}) + D(\mathcal{P}^{(2)} || \mathcal{Q}^{(2)})$ \\ \cline{2-3} 
				& Tsallis divergence & $D_q(\mathcal{P}^{(1)} \otimes \mathcal{P}^{(2)} || \mathcal{Q}^{(1)} \otimes \mathcal{Q}^{(2)}) = D_q(\mathcal{P}^{(1)} || \mathcal{Q}^{(1)}) + D_q(\mathcal{P}^{(2)} || \mathcal{Q}^{(2)}) - (q - 1) D_q(\mathcal{P}^{(1)} || \mathcal{Q}^{(1)}) D_q(\mathcal{P}^{(2)} || \mathcal{Q}^{(2)})$ \cite{furuichi2004fundamental} \\ \cline{2-3}
				& $D_{\{k , r\}}$ & $D_{\{k,r\}}(\mathcal{P}^{(1)} \otimes \mathcal{P}^{(2)} || \mathcal{Q}^{(1)} \otimes \mathcal{Q}^{(2)}) = D_{\{k,r\}}(\mathcal{P}^{(1)} || \mathcal{Q}^{(1)}) + D_{\{k,r\}}(\mathcal{P}^{(2)} || \mathcal{Q}^{(2)}) - 2k D_{\{k,r\}}(\mathcal{P}^{(1)} || \mathcal{Q}^{(1)}) D_{\{k,r\}}(\mathcal{P}^{(2)} || \mathcal{Q}^{(2)})$ (Theorem \ref{pseudo_additivity_of_divergence})\\ \hline  \hline
				\multirow{3}{=}{Joint-convexity:  Let $\mathcal{P}^{(k)} = \{p^{(k)}(x)\}_{x \in X}$ and $\mathcal{Q}^{(k)} = \{q^{(k)}(x)\}_{x \in X}$ for $k = 1, 2$ are probability distributions. Construct new probability distributions $(1 - \lambda)\mathcal{P}^{(1)} + \lambda \mathcal{P}^{(2)} = \{(1 - \lambda)p^{(1)}(x) + \lambda p^{(2)}(x)\}_{x \in X}$, and $(1 - \lambda)\mathcal{Q}^{(1)} + \lambda \mathcal{Q}^{(2)} = \{(1 - \lambda)q^{(1)}(x) + \lambda q^{(2)}(x)\}_{x \in X}$ as convex combinations.} & KL divergence & $D((1 - \lambda)\mathcal{P}^{(1)} + \lambda \mathcal{P}^{(2)} || (1 - \lambda) \mathcal{Q}^{(1)} + \lambda \mathcal{Q}^{(2)}) \leq (1 - \lambda) D(\mathcal{P}^{(1)} || \mathcal{Q}^{(1)}) + \lambda D(\mathcal{P}^{(2)} || \mathcal{Q}^{(2)})$ \\ \cline{2-3} 
				& Tsallis divergence & $D_q((1 - \lambda)\mathcal{P}^{(1)} + \lambda \mathcal{P}^{(2)} || (1 - \lambda) \mathcal{Q}^{(1)} + \lambda \mathcal{Q}^{(2)}) \leq (1 - \lambda) D_q(\mathcal{P}^{(1)} || \mathcal{Q}^{(1)}) + \lambda D_q(\mathcal{P}^{(2)} || \mathcal{Q}^{(2)})$ \cite{furuichi2004fundamental} \\ \cline{2-3}
				& $D_{\{k , r\}}$ & $D_{\{k, r\}}((1 - \lambda)\mathcal{P}^{(1)} + \lambda \mathcal{P}^{(2)} || (1 - \lambda) \mathcal{Q}^{(1)} + \lambda \mathcal{Q}^{(2)}) \leq (1 - \lambda) D_{\{k, r\}}(\mathcal{P}^{(1)} || \mathcal{Q}^{(1)}) + \lambda D_{\{k, r\}}(\mathcal{P}^{(2)} || \mathcal{Q}^{(2)})$ (Theorem \ref{joint_convexity_of_divergence} ) \vspace{3cm} \\ \hline
			\end{tabular}
		\end{table}

	\section*{Acknowledgments}

		S.D. was a Post Doctoral Research Associate-1 at the S. N. Bose National Centre for Basic Sciences during this work. He is also thankful to Antonio Maria Scarfone and Bibhas Adhikari for some suggestions and carefully revising the manuscript. S.F. was partially supported by JSPS KAKENHI Grant Number 16K05257.


\begin{thebibliography}{10}
	
	\bibitem{tsallis1988possible}
	Constantino Tsallis.
	\newblock {Possible generalization of Boltzmann-Gibbs statistics}.
	\newblock {\em Journal of statistical physics}, 52(1-2):479--487, 1988.
	
	\bibitem{de2004image}
	M~Portes De~Albuquerque, Israel~A Esquef, and AR~Gesualdi Mello.
	\newblock {Image thresholding using Tsallis entropy}.
	\newblock {\em Pattern Recognition Letters}, 25(9):1059--1065, 2004.
	
	\bibitem{zhang2009application}
	Dandan Zhang, Xiaofeng Jia, Haiyan Ding, Datian Ye, and Nitish~V Thakor.
	\newblock {Application of Tsallis entropy to EEG: quantifying the presence of
		burst suppression after asphyxial cardiac arrest in rats}.
	\newblock {\em IEEE transactions on biomedical engineering}, 57(4):867--874,
	2009.
	
	\bibitem{chen2014tsallis}
	Jikai Chen and Guoqing Li.
	\newblock Tsallis wavelet entropy and its application in power signal analysis.
	\newblock {\em Entropy}, 16(6):3009--3025, 2014.
	
	\bibitem{becker2019convergence}
	Simon Becker and Nilanjana Datta.
	\newblock Convergence rates for quantum evolution and entropic continuity
	bounds in infinite dimensions.
	\newblock {\em Communications in Mathematical Physics}, pages 1--49, 2019.
	
	\bibitem{abe2002towards}
	Sumiyoshi Abe and AK~Rajagopal.
	\newblock Towards nonadditive quantum information theory.
	\newblock {\em Chaos, Solitons \& Fractals}, 13(3):431--435, 2002.
	
	\bibitem{sharma1975entropy}
	Bhu~D Sharma and Inder~J Taneja.
	\newblock Entropy of type ($\alpha$, $\beta$) and other generalized measures in
	information theory.
	\newblock {\em Metrika}, 22(1):205--215, 1975.
	
	\bibitem{mittal1975some}
	DP~Mittal.
	\newblock On some functional equations concerning entropy, directed divergence
	and inaccuracy.
	\newblock {\em Metrika}, 22(1):35--45, 1975.
	
	\bibitem{frank2000exact}
	TD~Frank and A~Daffertshofer.
	\newblock {Exact time-dependent solutions of the Renyi Fokker--Planck equation
		and the Fokker--Planck equations related to the entropies proposed by Sharma
		and Mittal}.
	\newblock {\em Physica A: Statistical Mechanics and its Applications},
	285(3-4):351--366, 2000.
	
	\bibitem{paul2016sharma}
	Jerin Paul and Poruthiyudian~Yageen Thomas.
	\newblock {Sharma-Mittal entropy properties on record values}.
	\newblock {\em Statistica}, 76(3):273--287, 2016.
	
	\bibitem{koltcov2019estimating}
	Sergei Koltcov, Vera Ignatenko, and Olessia Koltsova.
	\newblock {Estimating topic modeling performance with Sharma--Mittal Entropy}.
	\newblock {\em Entropy}, 21(7):660, 2019.
	
	\bibitem{crupi2018generalized}
	Vincenzo Crupi, Jonathan~D Nelson, Bj{\"o}rn Meder, Gustavo Cevolani, and Katya
	Tentori.
	\newblock Generalized information theory meets human cognition: Introducing a
	unified framework to model uncertainty and information search.
	\newblock {\em Cognitive Science}, 42(5):1410--1456, 2018.
	
	\bibitem{jahromi2018generalized}
	A~Sayahian Jahromi, SA~Moosavi, H~Moradpour, JP~Morais Gra{\c{c}}a, IP~Lobo,
	IG~Salako, and A~Jawad.
	\newblock Generalized entropy formalism and a new holographic dark energy
	model.
	\newblock {\em Physics Letters B}, 780:21--24, 2018.
	
	\bibitem{younas2019cosmological}
	M~Younas, Abdul Jawad, Saba Qummer, H~Moradpour, and Shamaila Rani.
	\newblock {Cosmological implications of the generalized entropy based
		holographic dark energy models in dynamical Chern-Simons modified gravity}.
	\newblock {\em Advances in High Energy Physics}, 2019, 2019.
	
	\bibitem{sadeghi2019investigation}
	J~Sadeghi, M~Rostami, and MR~Alipour.
	\newblock {Investigation of phase transition of BTZ black hole with
		Sharma--Mittal entropy approaches}.
	\newblock {\em International Journal of Modern Physics A}, 34(30):1950182,
	2019.
	
	\bibitem{ghaffari2019black}
	S~Ghaffari, AH~Ziaie, H~Moradpour, F~Asghariyan, F~Feleppa, and M~Tavayef.
	\newblock {Black hole thermodynamics in Sharma--Mittal generalized entropy
		formalism}.
	\newblock {\em General Relativity and Gravitation}, 51(7):93, 2019.
	
	\bibitem{borges1998family}
	Ernesto~P Borges and Itzhak Roditi.
	\newblock A family of nonextensive entropies.
	\newblock Technical report, SCAN-9905035, 1998.
	
	\bibitem{kaniadakis2004deformed}
	G~Kaniadakis, M~Lissia, and AM~Scarfone.
	\newblock Deformed logarithms and entropies.
	\newblock {\em Physica A: Statistical Mechanics and its Applications},
	340(1-3):41--49, 2004.
	
	\bibitem{kaniadakis2005two}
	G~Kaniadakis, M~Lissia, and AM~Scarfone.
	\newblock Two-parameter deformations of logarithm, exponential, and entropy: A
	consistent framework for generalized statistical mechanics.
	\newblock {\em Physical Review E}, 71(4):046128, 2005.
	
	\bibitem{naudts2002deformed}
	Jan Naudts.
	\newblock Deformed exponentials and logarithms in generalized thermostatistics.
	\newblock {\em Physica A: Statistical Mechanics and its Applications},
	316(1-4):323--334, 2002.
	
	\bibitem{wada2007two}
	Tatsuaki Wada and Hiroki Suyari.
	\newblock {A two-parameter generalization of Shannon--Khinchin axioms and the
		uniqueness theorem}.
	\newblock {\em Physics Letters A}, 368(3-4):199--205, 2007.
	
	\bibitem{furuichi2010axiomatic}
	Shigeru Furuichi.
	\newblock An axiomatic characterization of a two-parameter extended relative
	entropy.
	\newblock {\em Journal of Mathematical Physics}, 51(12):123302, 2010.
	
	\bibitem{furuichi2004fundamental}
	Shigeru Furuichi, Kenjiro Yanagi, and Ken Kuriyama.
	\newblock {Fundamental properties of Tsallis relative entropy}.
	\newblock {\em Journal of Mathematical Physics}, 45(12):4868--4877, 2004.
	
	\bibitem{furuichi2006information}
	Shigeru Furuichi.
	\newblock {Information theoretical properties of Tsallis entropies}.
	\newblock {\em Journal of Mathematical Physics}, 47(2):023302, 2006.
	
	\bibitem{furuichi2005uniqueness}
	Shigeru Furuichi.
	\newblock {On uniqueness theorems for Tsallis entropy and Tsallis relative
		entropy}.
	\newblock {\em IEEE Transactions on Information Theory}, 51(10):3638--3645,
	2005.
	
	\bibitem{boyd2004convex}
	Stephen Boyd and Lieven Vandenberghe.
	\newblock {\em Convex optimization}.
	\newblock Cambridge university press, 2004.
	
	\bibitem{cover2012elements}
	Thomas~M Cover and Joy~A Thomas.
	\newblock {\em Elements of information theory}.
	\newblock John Wiley \& Sons, 2012.
	
	\bibitem{stewart1995multivariable}
	James Stewart.
	\newblock {\em Multivariable Calculus}.
	\newblock Brooks/Cole, CA, 1995.
	
	\bibitem{amari2007methods}
	Shun-ichi Amari and Hiroshi Nagaoka.
	\newblock {\em Methods of information geometry}, volume 191.
	\newblock American Mathematical Soc., 2007.
	
	\bibitem{yamano2002some}
	Takuya Yamano.
	\newblock Some properties of q-logarithm and q-exponential functions in tsallis
	statistics.
	\newblock {\em Physica A: Statistical Mechanics and its Applications},
	305(3-4):486--496, 2002.
	
\end{thebibliography}

\end{document}